\documentclass{llncs}

\usepackage{lipsum}
\usepackage{amsfonts}
\usepackage{graphicx}
\usepackage{epstopdf}
\usepackage{algorithmic}
\usepackage{amssymb}

\ifpdf
  \DeclareGraphicsExtensions{.eps,.pdf,.png,.jpg}
\else
  \DeclareGraphicsExtensions{.eps}
\fi

\newcommand{\TheTitle}{Circular Backbone Colorings: on matching and tree backbones of planar graphs} 
\newcommand{\TheAuthors}{J. Araujo, F. Benevides, A. Cezar and A. Silva}

\usepackage{amsopn}

\usepackage[english]{babel}
\usepackage{graphicx}
\usepackage{amssymb, amsmath}

\title{\TheTitle}
\author{\TheAuthors}
\institute{ParGO Group - Parallelism, Graphs and Optimization\\Universidade Federal do Cear\'a, Fortaleza, Brazil}

\usepackage{multirow}
\usepackage{amssymb,amsmath,amsfonts}
\usepackage{epsfig}
\usepackage{color}
\usepackage{url}
\usepackage{enumerate}

\newtheorem{Drule}{Rule}
\newenvironment{claimproof}[1]{\par\noindent\underline{Proof:}\space#1}{\hfill $\blacksquare$}

\input amssym.def

\newcommand{\Z}{\mathbb{Z}}

\newcommand{\Li}[1]{\mathcal{L}(#1)}
\newcommand{\Lip}[1]{\mathcal{L}'(#1)}
\newcommand{\Lp}{\mathcal{L}'}
\newcommand{\Lcal}{\mathcal{L}}
\newcommand{\cbeta}{\lvert\beta\rvert}
\newcommand{\adj}[1]{\langle #1\rangle}

\begin{document}

\maketitle

\begin{abstract}
%
%
Given a graph $G$, and a spanning subgraph $H$ of $G$, a circular $q$-backbone $k$-coloring of $(G,H)$ is a proper $k$-coloring $c$ of $G$ such that $q\le \lvert c(u)-c(v)\rvert \le k-q$, for every edge $uv\in E(H)$. The circular $q$-backbone chromatic number of $(G,H)$, denoted by $CBC_q(G,H)$, is the minimum integer $k$ for which there exists a circular $q$-backbone $k$-coloring of $(G,H)$.
  The Four Color Theorem implies that whenever $G$ is planar, we have $CBC_2(G,H)\le 8$. It is conjectured that this upper bound can be improved to 7 when $H$ is a tree, and to 6 when $H$ is a matching. In this work, we show that: 
1) if $G$ is planar and has no $C_4$ as subgraph, and $H$ is a linear spanning forest of $G$, then $CBC_2(G,H)\leq 7$; 2) if $G$ is a plane graph having no two 3-faces sharing an edge, and $H$ is a matching of $G$, then $CBC_2(G,H)\leq 6$; and 3) if $G$ is  planar and has no $C_4$ nor $C_5$ as subgraph, and $H$ is a mathing of $G$, then $CBC_2(G,H)\leq 5$. 
These results partially answers questions posed by Broersma, Fujisawa and Yoshimoto (2003), and by Broersma, Fomin and Golovach (2007). It also points towards a positive answer for the Steinberg's Conjecture.
\end{abstract}




\section{Introduction}

For basic notions and terminology on Graph Theory, the reader is referred to~\cite{BM.book}. In this text, we only consider simple graphs. 

Let $G = (V,E)$ be a graph. A \emph{(proper) $k$-coloring of $G$} is a function $c:V(G)\rightarrow \{1,\ldots,k\}$ such that $c(u)\neq c(v)$, for every edge $uv\in E(G)$. $G$ is \emph{$k$-colorable} if there exists a $k$-coloring of $G$. The \emph{chromatic number} of $G$, denoted by $\chi(G)$, is the smallest $k$ for which $G$ has a $k$-coloring. $G$ is \emph{$k$-chromatic} if $\chi(G)= k$. The {\sc Vertex Coloring Problem} consists in determining $\chi(G)$, for a given graph $G$.

Among many practical problems that can be modeled using graph coloring, Frequency Assignment problems are perhaps the most famous ones~\cite{AHK+07}. There are several variations of the {\sc Vertex Coloring Problem} that were defined in order to model the specific constraints of the practical applications related to frequency assignment in networks.
The {\sc Backbone Coloring Problem} was defined by Broersma et al.~\cite{Broersma2003, BFGW07} in the context of Frequency Assigment Problems where certain channels of communication were more demanding than others.

Formally, given a graph $G$, a spanning subgraph $H$ of $G$, called the \emph{backbone} of G, and two positive integers $q$ and $k$, a \emph{$q$-backbone $k$-coloring of $(G,H)$} is a $k$-coloring $c$ of $G$ for which $\lvert c(u)-c(v)\rvert \ge q$, for every $uv\in E(H)$. The \emph{$q$-backbone chromatic number of $(G,H)$}, denoted by $BBC_q(G,H)$, is the minimum $k$ for which there exists a $q$-backbone $k$-coloring of $(G,H)$. The {\sc Backbone Coloring Problem} consists in determining $BBC_q(G,H)$. In this work, we focus on the case $q=2$ and thus we usually omit $q$ from the notation.

In their seminal article, Broersma et al. observe that 
\begin{equation}
\label{eq:bbc_chi}
BBC(G,H) \le 2\cdot\chi(G)-1.  
\end{equation}
This can be easily seen by considering an optimal coloring of $G$ that uses only odd colors. Note that, thanks to the Four Color Theorem~\cite{appel1977,appel1977a}, whenever $G$ is a planar graph and $H$ is any spanning subgraph of $G$, we get an upper bound of~7 to the backbone chromatic number of $(G,H)$. However, when $H$ is a spanning tree of $G$, Broersma et al. conjecture that this upper bound is in fact~6, and they show that this would be best possible~\cite{BFGW07}. 
\begin{conjecture}[\cite{BFGW07}]
\label{conj:Broersma}
If $G$ is a planar graph and $T$ is a spanning tree of $G$, then $$BBC(G,T)\le 6.$$
\end{conjecture}

In the literature, the only result approaching directly this conjecture shows that it holds whenever $T$ has diameter at most~4~\cite{CHRS.13}.

The authors in~\cite{HK.12,HKLT.12,HKLT.14} consider special backbone $k$-colorings where the color space is ``circular'', i.e., it behaves as $\Z/k$. More formally, given a graph $G$, a spanning subgraph $H$ of $G$, and a positive integer $q$, a \emph{circular $q$-backbone $k$-coloring of $(G,H)$} is a function $c:V(G)\rightarrow \{1,\ldots,k\}$ such that $q\le\lvert c(u)-c(v)\rvert\le k-q$, for every $uv\in E(H)$. The \emph{circular $q$-backbone chromatic number} of $(G,H)$, denoted by $CBC_q(G,H)$, is the smallest $k$ for which there exists a circular $q$-backbone $k$-coloring of $(G,H)$. Once more, we quite often omit the index $q$ whenever $q=2$. In order to simplify the notation, we often write CBC-$k$-coloring instead of circular 2-backbone $k$-coloring.

Note that any CBC-$k$-coloring of $(G,H)$ is also a backbone $k$-coloring of $(G,H)$, and, conversely, if $c$ is a backbone $k$-coloring of $(G,H)$, then it can also be seen as a CBC-$(k+1)$-coloring of $(G,H)$. Therefore we get:
\begin{equation}\label{eq:BBC_CBC}BBC(G,H)\le CBC(G,H)\le BBC(G,H) + 1.\end{equation}

Consequently, as far as Conjecture~\ref{conj:Broersma} is not proved to be true, then the following circular version of it is also opened: 

\begin{conjecture}
\label{conj:weakBroersma}
If $G$ is a planar graph and $T$ is a spanning tree of $G$, then $$CBC(G,T)\le 7.$$
\end{conjecture}

One may observe that a graph $G$ whose chromatic number is $\chi(G)=k$, satisfies $CBC_2(G,H)\leq 2k$, by combining Inequalities~\ref{eq:bbc_chi} and~\ref{eq:BBC_CBC}. Steinberg conjectures that every planar graph $G$ having no $C_4$ or $C_5$ as subgraph satisfies $\chi(G)\leq 3$~\cite{Ste93}. Consequently, one may wonder whether:

\begin{conjecture}
\label{conj:weakSteinberg}
If $G$ is a planar graph having no $C_4$ or $C_5$ as subgraph, then $CBC_2(G,H)\leq 6$, for every backbone $H\subseteq G$.
\end{conjecture}
Notice that Conjecture~\ref{conj:weakSteinberg} is in fact equivalent to Steinberg's Conjecture when $H=G$.

In this paper, we prove particular cases of Conjectures~\ref{conj:weakBroersma} and~\ref{conj:weakSteinberg}.

\subsection{Matching Backbones} It is known that if $G$ is a 3-colorable graph and $M$ is a matching of $G$, then $BBC(G,M)\le 4$~\cite{BFY.03}. Combining this result with Inequality \ref{eq:BBC_CBC}, we observe that if Steinberg's Conjecture is true, then $CBC(G,M)\le 5$, whenever $G$ is a planar graph without cycles of length~4 or~5, and $M$ is a matching of $G$. We first prove that this bound holds, giving yet more evidence to the validity of Steinberg's Conjecture:

\begin{theorem}\label{thm:noC4C5}
If $G$ is a planar graph without cycles of length~4 or~5 as subgraph, and $M$ is a matching of $G$, then $CBC(G,M)\le 5$.
\end{theorem}

In~\cite{BFY.03} the authors prove that $BBC(G,M)\le 6$, whenever $G$ is a planar graph and $M$ is a matching. They also ask whether $BBC(G,M)\le 5$ holds, and whether $BBC(G,M)\le 6$ can be proved without using the Four Color Theorem. We partially answer both questions by showing that:

\begin{theorem}\label{thm:noadjacentC3}
If $G$ is a plane graph with no two faces of degree~3 that share an edge, and $M$ is a matching in $G$, then $CBC(G,M)\le 6$.
\end{theorem}

Although our result restricts the class of graphs when compared to the result presented in~\cite{BFY.03}, it is stronger on this restricted class since we deal with \emph{circular} backbone colorings instead. We mention that our result points to a positive answer to the question about whether $BBC(G,M)\le 5$, and that our proof does not use the Four Color Theorem.

\subsection{Linear Forest Backbones} Finally, we also study more general backbones. A forest is called \emph{linear} if its components are paths.

In~\cite{AHS.15}, the authors investigate $CBC(G,F)$ in the light of Steinberg's Conjecture~\cite{Ste93}. Araujo et al. prove that if $G$ is a planar graph with no cycles of length 4 or 5, then $CBC(G,F)\le 7$ whenever $F$ is a spanning forest of $G$, and that $CBC(G,F)\le 6$, whenever $F$ is a spanning linear forest of $G$~\cite{AHS.15}.  Observe that their results partially solve Conjectures~\ref{conj:weakBroersma} and~\ref{conj:weakSteinberg}. 

The last result we present in this work is similar to theirs by considering planar graphs with no cycles of length~4 and linear forests as backbones.

\begin{theorem}\label{main:forest}
If $G$ is a planar graph without cycles of length 4 as subgraph, and $F$ is a linear spanning forest of $G$, then $CBC(G,F)\le 7$. 
\end{theorem}
Although in our proof we can consider graphs that have $C_5$ as subgraph, we need an extra color than in the previous result in the literature. However, this was expected since our efforts were done towards an answer to Conjecture~\ref{conj:weakBroersma}.

The remainder of this text is organized as follows: in Section~\ref{sec:defsandlemma}, we introduce basic notation and results. Then, we prove Theorems~\ref{thm:noC4C5},~\ref{thm:noadjacentC3} and~\ref{main:forest} in Sections~\ref{sec:matchings_C4C5},~\ref{sec:noadjacentC3} and~\ref{sec:forest}, respectively. 


\section{Preliminaries}
\label{sec:defsandlemma}
For the basic definitions about simple graphs and planar graphs, we refer the reader once again to~\cite{BM.book}.

Given a statement $P$, and a partially ordered set $(\mathcal{S}, \preceq)$, we denote by $P(\mathcal{S})$ the set $\{S\in \mathcal{S}\mid P \mbox{ holds for }S\}$. And we say that $S\in \mathcal{S}$ is a \emph{minimal counterexample for $P$} if $S\notin P(\mathcal{S})$, and $S'\in P(\mathcal{S})$ for every $S'\in \mathcal{S}$ such that $S'\prec S$. In our proofs, we consider minimal counterexamples to our theorems. For this, we consider a pair $(G',H')$ to be smaller than a pair $(G,H)$ if $G'\subset G$ and $H'\subseteq H$; in this case we say that $(G',H')$ is a \emph{subpair} of $(G,H)$.

In what follows, given a minimal counterexample $(G,H)$ to one of our theorems, we get a contradiction by being successful in extending a partial CBC-$k$-coloring of $(G',H')$ to $(G,H)$, where $(G',H')$ is a subpair of $(G,H)$. The following lemma presented in~\cite{AHS.15} will be useful. It can be easily proved by considering a CBC-$k$-coloring of $(G-u,H-u)$ and observing that it can be extended to a CBC-$k$-coloring of $(G,H)$.

\begin{lemma}[\cite{AHS.15}]\label{lem:totaldegree}
If $(G,H)$ is minimal such that $CBC(G,H)>k$, then, for every $u\in V(G)$, we have that $d_G(u)+2d_H(u)\ge k$, .
\end{lemma}

The general technique used to prove the above lemma is also extensively applied in the remainder of the text. Because of this, we introduce the following definitions and notation. 

Given a positive integer $k$, we denote the set $\{1,\cdots,k\}$ by $[k]$, and given $c\in [k]$, we denote by $\adj{c}$ the set $\{d\in [k]\mid \lvert c-d\rvert \le 1\mbox{ or } \lvert c-d\rvert \ge k-1\}$ (the colors adjacent to $c$ in the circular space $[k]$). Also, we denote the power set of $[k]$ by $2^{[k]}$.
Given a pair $(G,H)$, a subgraph $G'\subset G$, and a CBC-$k$-coloring $\psi$ of $(G',H[V(G')])$, we define, for each $u\in V(G)\setminus V(G')$, the set of \emph{available colors for $u$ in $\psi$}:
\[A_\psi(u) = [k]\setminus (\psi(N_{G'}(u))\cup\{\adj{\psi(v)}\mid v\in N_{H'}(u)\}).\]
Also, we denote $\lvert A_\psi(u)\rvert$ by $a_\psi(u)$.


\section{Proof of Theorem~\ref{thm:noC4C5}}\label{sec:matchings_C4C5}

In order to prove Theorem~\ref{thm:noC4C5}, we need the following lemma, proved in~\cite{AHS.15}.

\begin{lemma}[\cite{AHS.15}]\label{lem:AHS.15}
Let $G$ be a plane graph without cycles of length~4 or~5, $G\neq K_3$, and let $n$ and $f_3$ denote the number of vertices of $G$ and number of faces of degree~3 in $G$, respectively. Then, \[\sum_{v\in V(G)}d(v)\le 3n+\frac{3f_3}{2}-6.\]
\end{lemma}

We use the discharging method to prove that if $(G,M)$ is a minimal couterexample to Theorem~\ref{thm:noC4C5}, then Lemma~\ref{lem:AHS.15} does not hold for $G$. This means that no counterexample can exist and that the theorem holds. The following lemma will be useful.

\begin{lemma}\label{lem:minDegree}
Let $(G,M)$ be a minimal counterexample to Theorem~\ref{thm:noC4C5}. Then, we have $\delta(G)\ge 3$. Furthermore,  if $u\in V(G)$ has degree~3, then $u$ is incident to some edge in $M$, say $uw$, and $w$ is such that $d(w)\ge 5$.
\end{lemma}
\begin{proof}
Let $u\in V(G)$, and denote by $T$ the subgraph $(V(G),M)$. By Lemma~\ref{lem:totaldegree} and because $d_T(u)\le 1$, we get that $d_G(u)\ge 3$. Similarly, if $d_G(u) \le 4$, we must have $d_T(u) = 1$. So, suppose that $u\in V(G)$ has degree~3 and let $w\in V(G)$ be such that $uw\in M$. By contradiction, suppose that $d(w)\le 4$, and let $\psi$ be a CBC-5-coloring of $(G-u-w, M-uw)$. Note that $a_\psi(u)\ge 3$ and $a_\psi(w)\ge 2$. Therefore, there exists a color $c\in A_\psi(w)$ such that $A_\psi(u)\setminus \adj{c}\neq \emptyset$. This implies that $\psi$ can be extended to $(G,M)$, a contradiction.
\end{proof}

Denote by $F_3$ the set of faces of degree~3 of $G$. We start by giving charge $d(v)-3$ for every $v\in V(G)$, and $-\frac{3}{2}$ for every $t\in F_3$. We want to distribute the charge between the vertices of $G$ and the faces in $F_3$ in such a way as to ensure that at the end, each vertex and each face in $F_3$ has nonnegative charge. Because the total amount of charge does not change, we get (below, $f_3$ and $n$ represent $\lvert F_3\rvert$ and $\lvert V(G)\rvert$, respectively):
\[\sum_{v\in V(G)}(d(v)-3)-\frac{3f_3}{2} \ge 0\Leftrightarrow \sum_{v\in V(G)}d(v) \ge 3n + \frac{3f_3}{2}.\]
This contradicts Lemma~\ref{lem:AHS.15}. To prove this can be done, we apply the following discharging rules. Below, given $u\in V(G)$, we denote by $F_3(u)$ the set of faces of degree~3 containing $u$.

\begin{Drule}\label{noC4C5R1}
For each $uw\in M$ such that $d(u) = 3$, send $\frac{1}{2}$ charge from $w$ to $u$.
\end{Drule}

\begin{Drule}\label{noC4C5R2}
For each $u\in V(G)$ and each $t\in F_3(u)$, send charge $\frac{1}{2}$ from $u$ to $t$.
\end{Drule}

\begin{proof}[of Theorem~\ref{thm:noC4C5}]
For each $x\in V(G)\cup F_3$, denote by $\mu_0(x),\mu_1(x),\mu_2(x)$ the charge of $x$ before Rule~\ref{noC4C5R1} has been applied, before Rule~\ref{noC4C5R2} has been applied and after Rule~\ref{noC4C5R2} has been applied, respectively. Because $M$ is a matching, no vertex is incident to more than one edge in $M$. Thus, by Lemma \ref{lem:minDegree}, we get the following:
\begin{itemize}
\item If $d(u) = 3$, then $\mu_1(u) = \frac{1}{2}$; 
\item If $d(u) = 4$, then $\mu_1(u) = \mu_0(u) = 1$; and
\item If $d(u) \ge 5$, then $\mu_1(u) \ge \mu_0(u) - \frac{1}{2} = \frac{2d(u)-7}{2}$.
\end{itemize}
Now, for each $u\in V(G)$, denote by $f_3(u)$ the value $\lvert F_3(u)\rvert$. Note that, since $G$ has no cycles of length~4, no two faces in $F_3$ can share an edge. This implies that $f_3(u)\le \lfloor \frac{d(u)}{2}\rfloor$. One can verify by what is said above that $\mu_1(u)\ge \frac{d(u)}{4}\ge \frac{f_3(u)}{2}$. This means that after distributing charge $1/2$ to each $t\in F_3(u)$, we get that $u$ still has non-negative charge, i.e, $\mu_2(u)\ge 0$ for every $u\in V(G)$. Finally, because each $t\in F_3$ receives charge $1/2$ from each vertex in $t$, we get $\mu_2(t) = \mu_0(t) + 3/2 = 0$.
\end{proof}


\section{Proof of Theorem~\ref{thm:noadjacentC3}}\label{sec:noadjacentC3}

Consider a plane graph $G$ and its dual $G^*$, and let $F_3$ be the set of faces of degree~3 in $G$ (alternatively, the set of vertices of degree~3 in $G^*$). We denote the graph $G^*-F_3$ by $G^*_4$, and say that a component of $G^*_4$ is an \emph{island of $G$}. Also, if $H$ is an acyclic component of $G^*_4$ such that $d_{G^*}(f) = 4$, for every $f\in V(H)$, then we say that $H$ is a \emph{bad island of $G$}. 
We denote the set of bad islands of $G$ by $\Gamma$ and we let $\gamma$ denote $\lvert \Gamma\rvert$. Let $f\in F_3$ and $H$ be an island of $G$; we say that $f$ \emph{share an edge with $H$} if $N_H(f)\neq \emptyset$ (i.e., if $f$ and $f'$ share an edge in $G$ for some $f'\in V(H)$). Also, we denote by $\Gamma(f)$ the set of bad islands that share an edge with~$f$.

\begin{lemma}\label{mainlemma_mat}
Let $G$ be a plane graph with no two faces of degree~3 sharing an edge, and let $f_3$ denote the number of faces of degree~3 in $G$. Then, \[\sum_{v\in V(G)}d(v) \le 5\lvert V(G)\rvert + \gamma - f_3 - 10.\]
\end{lemma}
\begin{proof}
Let $f_4$ denote the number of faces of degree~4 in $G$, $\lvert E(G)\rvert$ be denoted by $m$, ${\cal F}$ denote the set of faces of $G$ and, given $f\in {\cal F}$, let $\lvert f\rvert$ denote the degree of $f$. We claim that: 

\begin{equation}\label{claim1} 3f_3+f_4\le m+\gamma\end{equation}

This implies that $\sum_{f\in {\cal F}}(\lvert f\rvert - 5) \ge -2f_3-f_4\ge -m-\gamma+f_3$. On the other hand $\sum_{f\in {\cal F}}(\lvert f\rvert) - 5\lvert {\cal F}\rvert = 2m - 5\lvert {\cal F}\rvert$. Combining these and applying Euler's Formula we get (below, $n$ denotes $\lvert V(G)\rvert$):
$$ 2m - 5(2-n+m)\ge -m-\gamma +f_3 \Longleftrightarrow 2m\le 5n+\gamma - f_3 -10$$

It remains to prove Inequality~\ref{claim1}. For this, we partition $E(G)$ in $E_3,\overline{E}_3$, where $E_3$ is described below and $\overline{E}_3 = E(G)\setminus E_3$.
$$E_3 = \{e\in E(G)\mid \mbox{ $e$ is in the boundary of some face of degree 3}\}.$$

Because $G$ has no two faces of degree~3 sharing an edge, we get $\lvert E_3\rvert = 3f_3$. We prove that $\lvert \overline{E}_3\rvert \ge f_4 - \gamma$, thus finishing the proof. For this, note that if $e\in \overline{E}_3$, then there is an edge $e^*$ in $G^*_4$ related to $e$. On the other hand, if $e^*\in E(G^*_4)$, then $e^*$ is related to an edge $e\in E(G)$ that separates faces of degree at least~4; hence, $e\in \overline{E}_3$. Therefore, $\lvert \overline{E}_3\rvert = \lvert E(G^*_4)\rvert$. Finally, because the number of edges in any graph is at least the number of vertices minus the number of acyclic components of the graph, we get:
$$\lvert \overline{E}_3\rvert \ge \lvert V(G^*_4)\rvert - \gamma \ge f_4-\gamma.$$
$\Box$\end{proof}

Now, by supposing that there exists a couterexample $(G,M)$ to Theorem~\ref{thm:noadjacentC3}, we use the discharging method to get a contradiction to Lemma~\ref{mainlemma_mat}. For this, start by giving charge $d(v)-5$ to each $v\in V(G)$, charge 1 to each $f\in F_3$, and charge -1 to each $b\in \Gamma$. Then, we apply discharging rules and ensure that this initial charge can be redistributed in the graph in such a way that every vertex, every face of degree~3 and every bad island have non-negative charge. We get a contradiction since:
$$\sum_{v\in V(G)}(d(v)-5) + f_3 -\gamma \ge 0 \Leftrightarrow \sum_{v\in V(G)}d(v)\ge 5n + \gamma - f_3.$$

We need the following lemma. 

\begin{lemma}\label{lem:degree_mat}
Let $(G,M)$ be a minimal counterexample to Theorem~\ref{thm:noadjacentC3}. Then, we have $\delta(G)\ge 4$. Furthermore, if $u\in V(G)$ has degree~4, then $u$ is incident to an edge of $M$, say $uw$, and $d(w)\ge 6$.
\end{lemma}
\begin{proof}
Let $T$ denote the subgraph $(V(G),M)$. By Lemma~\ref{lem:totaldegree}, we get $\delta(G)\ge 4$, and that $d_T(u)=1$ whenever $d_G(u)\le 5$.
So, consider $u\in V(G)$ with degree~4, and suppose that $d(w)\le 5$, where $w$ is such that $uw\in M$. Let $\psi$ be a CBC-6-coloring of $(G-\{u,w\},M-\{u,w\})$. Then, $a_\psi(w)\ge 3$ and $a_\psi(u)\ge 2$. Therefore, there exists a color $c\in A_\psi(u)$ such that $A_\psi(w)\setminus \adj{c}\neq\emptyset$, which implies that $\psi$ can be extended to $(G,M)$, a contradiction.
\end{proof}

Let $V_4$ be the set of vertices with degree~4 in $G$, and for each $u\in V_4$, denote by $u^*$ the vertex such that $uu^*\in M$. The discharging rules are the following:

\setcounter{Drule}{0}
\begin{Drule}\label{MR1}
For each $f\in F_3$, send charge $\frac{1}{3}$ from $f$ to each $b\in \Gamma(f)$.
\end{Drule}

\begin{Drule}\label{MR2}
For each $u\in V_4$, send charge $1$ from $u^*$ to $u$.
\end{Drule}

\begin{proof}[of Theorem~\ref{thm:noadjacentC3}]
For each $x\in V(G)\cup F_3\cup \Gamma$, let $\mu_0(x),\mu_1(x),\mu_2(x)$ denote the charge of $x$ before Rule~\ref{MR1}, after Rule~\ref{MR1}, and after Rule~\ref{MR2} has been applied, respectively. Recall that $\mu_0(v) = d(v)-5$, for every $v\in V(G)$; $\mu_0(f) = 1$, for every $f\in F_3$; and $\mu_0(b) = -1$, for every $b\in \Gamma$. 

Because $M$ is a matching and by Lemma~\ref{lem:degree_mat}, we get that $\mu_2(v)\ge 0$, for every $v\in V(G)$. Also, for each $f\in F_3$, we have $\lvert \Gamma(f)\rvert\le 3$; hence $\mu_2(f) = \mu_1(f) = \mu_0(f) - \lvert \Gamma(f)\rvert/3 \ge 0$. It remains to prove that each bad island also ends up with non-negative charge. So, consider a bad island of $G$, i.e., an acyclic component $H$ of $G^*_4$ such that each  $f\in V(H)$ has degree~4 in $G^*$. If $V(H) = \{f\}$, because two faces of degree~3 in $G$ intersect in at most one vertex, we get that $f$ corresponds to an induced cycle of length~4 in $G$, which implies that $f$ is adjacent to~4 distinct vertices of $F_3$. And if $\lvert V(H)\rvert \ge 2$, then $H$ has at least one leaf, say $f$; as before, we get that $f$ is adjacent to at least~3 distinct vertices of $F_3$. In any case, we get that $y = \lvert \{f\in F_3\mid H\in \Gamma(f)\}\rvert \ge 3$, which implies that $\mu_2(H) = \mu_1(H) = \mu_0(H) + y/3\ge 0$.
\end{proof}


\section{Linear Forest Backbone}
\label{sec:forest}

We prove Theorem \ref{main:forest} in this section using the same general strategy, except that the structural properties needed are more complex. In the previous sections, a simple lemma concerning at most two vertices, say $u$ and $v$, was enough to say that a CBC-$k$-coloring $\psi$ of $(G-u-v,H-u-v)$ could be extended to $(G,H)$. Here, the backbone is a linear tree and therefore we sometimes need to remove entire subpaths from a minimal counterexample $(G,H)$. For this, we work with the lists $A_\psi$ in a more clever way. This is done in the next subsection.

\subsection{Forbidden Structures}

Let $(H,P)$ be such that $P\subseteq H$, $k$ be a positive integer, and  $\mathcal{L}:V(H)\rightarrow 2^{[k]}$. If there exists a CBC-$k$-coloring $\psi$ of $(H,P)$ such that $\psi(v)\in \mathcal{L}(v)$, for all $v\in V(H)$, then we say that $(H,P)$ is \emph{$\mathcal{L}$-CBC-$k$-colorable}. 
Throughout the proof, we sometimes consider $\mathcal{L}$ to be smallest possible in the context. This is not a problem since whenever $(H,P)$ is $\mathcal{L}$-CBC-$k$-colorable and $\mathcal{L}'$ is such that $\mathcal{L}(v)\subseteq \mathcal{L}'(v)$, for every $v\in V(H)$, we also have that $(H,P)$ is $\mathcal{L}'$-CBC-$k$-colorable. 

Consider a pair $(H,P)$ such that $P$ is a Hamiltonian path of $H$, and write $P$ as $(v_1,\ldots,v_n)$. Also, let
$\Lcal:V(H)\rightarrow 2^{[7]}$ be a list assignment for $H$, and $\Lp:V(H')\rightarrow 2^{[7]}$ be a list assignment for $H'\subseteq H$. We use the reduction rule below to prove the non-existence of certain structures in a minimal counterexample to Theorem~\ref{main:forest}. We denote the values $\lvert \Li{x}\rvert$ and $\lvert\mathcal{L}'(x)\rvert$ by $\ell(x)$ and $\ell'(x)$, respectively. 

  \textbf{Reduction Rule}: $((H',P'),\Lp)$ is a \emph{reduction of $((H,P),\Lcal)$ on $v_1$} if:
  \begin{itemize} 
  \item[-] $H'=H-v_1$;
  \item[-] $P' = P-v_1$;
  \item[-] $\ell'(v_2)\ge\ell(v_2)-2$;
  \item[-] $\ell'(x)\ge\ell(x)-1$, for every $x\in N(v_1)\setminus \{v_2\}$; 
  \item[-] $\ell'(x)=\ell(x)$, for every $x\in V(H)\setminus N[v_1]$; and 
  \item[-] If $\Li{v_2}\setminus \Lip{v_2}=\{c,d\}$, then $\lvert \adj{c}\cup \adj{d}\rvert \le 5$.
  \end{itemize}

We say that a reduction $((H',P'),\Lp)$ of $((H,P),\Lcal)$ on $v_1$ is \emph{extendable} if every $\Lp$-CBC-7-coloring of $(H',P')$ can be extended to an $\Lcal$-CBC-7-coloring of $(H,P)$. The following lemma gives suficient conditions for $((H,P),\Lcal)$ to have an extendable reduction.

\begin{lemma}\label{2reduction}
Let $H$ be any graph, $P=(v_1,\ldots,v_n)$ be a Hamiltonian path of $H$, and consider $\Lcal:V(H)\rightarrow 2^{[7]}$. If the conditions below hold, then $((H,P),\Lcal)$ has an extendable reduction on $v_1$.
\begin{enumerate}
\item \label{1} $d(v_1)\le 4$;
\item \label{2} $\ell(v_1)\ge 1+d(v_1)$; and
\item \label{3} If $d(v_1)=4$, and $c$ and $d$ are the colors not in $\Li{v_1}$, then $\lvert \adj{c}\cup \adj{d}\rvert \le 5$.
\end{enumerate}
\end{lemma}
\begin{proof}
Without loss of generality, suppose that $\ell(v_1) = 1+ d(v_1)$. First, suppose that $d(v_1)=1$. If $\Li{v_1}$ has two consecutive colors, then remove both from $\Li{v_2}$; if $\Li{v_1}=\{c-1,c+1\}$ for some $c\in [7]$, then remove $c$ from $\Li{v_2}$; otherwise, do not change $\Li{v_2}$. Let $\Lp$ be the obtained function. One can see that $((H-v_1,P-v_1),\Lp)$ is a reduction of $((H,P),\Lcal)$ on $v_1$. Let $\psi$ be an $\Lp$-CBC-$7$-coloring of $(H-v_1,P-v_1)$; if no such coloring exists, then the lemma holds by vacuity. By the choice of the removed colors, note that $\Li{v_1}\setminus\adj{\psi(v_2)}\neq \emptyset$, which means that $\psi$ can be extended to $v_1$. 

Now, consider $d(v_1)>1$. First, suppose that there exists $c\in\Li{v_1}$ such that $\{c-1,c+1\}\cap\Li{v_1}=\emptyset$. Let $\Lp$ be obtained by removing $c-1$ and $c+1$ from $\Li{v_2}$, and $c$ from $\Li{x}$, for every $x\in N(v_1)\setminus\{v_2\}$. Then $((H-v_1,P-v_1),\Lp)$ is a reduction of $((H,P),\Lcal)$ on $v_1$, and we want to show that it is extendable. So let $\psi$ be an $\mathcal{L}'$-CBC-$7$-coloring of $(H-v_1,P-v_1)$, and let $F = \psi(N(v_1))\cup\adj{\psi(v_2)}$, the set of colors that are forbidden for $v_1$. If $\psi(v_2)\neq c$, we can color $v_1$ with $c$. Otherwise, we get $\lvert \Li{v_1}\cap \adj{\psi(v_2)}\rvert =1$, which implies that 
\begin{equation}\label{Eq:2reduction}
\lvert \Li{v_1}\cap F\rvert \le \lvert \Li{v_1}\cap \psi(N(v_1)\setminus\{v_2\})\rvert + \lvert \Li{v_1}\cap \adj{\psi(v_2)}\rvert \le d(v_1)-1+1.
\end{equation} 
Since $\ell(v_1)=d(v_1)+1$, there is a color in $\Li{v_1}\setminus F$ with which we can color~$v_1$. 

Finally, suppose that: 
\begin{itemize}
\item[(*)] $\{c-1,c+1\}\cap \Li{v_1}\neq \emptyset$, for every $c\in\Li{v_1}$.
\end{itemize}

Because $2\le d(v_1)\le 4$ and $\ell(v_1)=d(v_1)+1$, we know that at least two colors are not in $\Li{v_1}$ and that $\ell(v_1)\ge 3$. Without loss of generality and by (*), we can suppose that $\{1,2\}\subset \Li{v_1}$ and that $7\notin \Li{v_1}$. We claim that we can also suppose that $6\notin\Li{v_1}$. Suppose otherwise; by (*) we get that $\{1,2,5,6\}\subseteq \Li{v_1}$. If $\{3,4\}\cap \Li{v_1}=\emptyset$, then we rotate the colors so that $1$ coincides with $5$ and the desired property holds. Otherwise, we get a contradiction to Property~\ref{3} since $\lvert \adj{c}\cup \adj{7}\rvert =6$ where $c\in \{3,4\}\setminus\Li{v_1}$.
Now, let $\Lp$ be obtained by removing 1 from $\Li{v_i}$, for every $v_i\in N(v_1)\setminus\{v_2\}$, and 1 and 2 from $\Li{v_2}$, and let $\psi$ be an $\mathcal{L}'$-CBC-$7$-coloring of $(H-v_1,P-v_1)$. If $\psi(v_2)\neq 7$, we can color $v_1$ with 1. Otherwise, since $\{6,7\}\cap \Li{v_1} = \emptyset$, we get $\lvert \adj{\psi(v_2)}\cap \Li{v_1}\rvert = 1$ and, again by Inequality~\ref{Eq:2reduction}, we get that there must exist a color in $\Li{v_1}$ with which we can color $v_1$.
\end{proof}

Now, we want to apply the above lemma to our problem. So, consider a planar graph $G$ with no cycles of length 4, a generating linear forest $F$ of $G$, and a subpath $P$ of $F$ with certain properties. If $(G,F)$ is a minimal counterexample to Theorem \ref{main:forest}, we know that there exists a CBC-7-coloring $\psi$ of $(G-P,F-P)$. We iteratively apply Lemma \ref{2reduction}, starting with $((G[V(P)],P),A_\psi)$, until we end up with a single vertex with list of size at least one. This implies that there exists an $A_\psi$-CBC-7-coloring $\psi'$ of $(G[V(P)],P)$, which in turn implies that $\psi$ can be extended to a CBC-7-coloring of $(G,F)$, thus contradicting the choice of $(G,F)$. This ensures the non-existence of such a path $P$ in a minimal counterexample. Before we present the types of paths that cannot occur in a minimal couterexample, we need a further definition.

Let $(G,F)$ be as in the previous paragraph, and $T$ be a component of $F$. If $P$ is a maximal subpath of $T$ containing only vertices of degree at most~5 in $G$, we say that $P$ is a \emph{heavy subpath of $T$}. 
The next lemma follows easily from Lemma~\ref{lem:totaldegree} and the fact that $F$ is a linear forest.

\begin{lemma}\label{lem:degree}
Let $(G,F)$ be a minimal counterexample to Theorem~\ref{main:forest}. Then, we have $\delta(G)\ge 3$, and if $v\in V(G)$ is such that $d_G(v)\le 4$, then $d_F(v)= 2$.
\end{lemma}

\begin{lemma}\label{lem:forbiddenStructures}
Let $(G,F)$ be a minimal counterexample to Theorem \ref{main:forest}, and  $P$ be a heavy subpath of a component of $F$. The following hold.
\begin{enumerate}[(a)]
\item \label{FS1} If $P$ has one vertex $v$ of degree 3, then $d(u)=5$, $\forall u\in V(P)\setminus \{v\}$; 
\item \label{FS3} If $P$ contains a leaf of $F$, then $d(u)= 5$, $\forall u\in V(P)$; and
\item \label{FS2} $P$ has at most two vertices of degree 4.
\end{enumerate}
\end{lemma}
\begin{proof}
Below, we consider a subpath $P'$ of $P$, and denote by $H$ the subgraph $G[V(P')]$. We prove that whenever $P$ does not satisfy one of the assertions, then, letting $\psi$ be a CBC-7-coloring of $(G-H,F-H)$, we get that $(H,P')$ is $A_\psi$-CBC-7-colorable, contradicting the fact that $(G,F)$ is a minimal counterexample to Theorem~\ref{main:forest}. We recall that, by Lemma~\ref{lem:degree}, we have $\delta(G)\ge 3$ and $d_G(u)\ge 5$ whenever $d_F(u)\le 1$.

First, suppose that either ($\ref{FS1}$) or ($\ref{FS3}$) does not hold, and let $P'=(v_1,v_2,\ldots,v_q)$ be a shortest subpath of $P$ such that $q\ge 2$, $d(v_1)\le 4$, and either $d(v_q)= 3$ or $v_q$ is a leaf in $P$. Also, let $\psi$ be a CBC-7-coloring of $(G-H,F-H)$. We construct a sequence $R_1,\ldots,R_q$ such that $R_1=((H,P'),A_\psi)$; $R_i$ is an extendable reduction of $R_{i-1}$ on $v_{i-1}$, for each $i\in \{2,\ldots,q\}$; and the list available for $v_q$ in $R_q$, say $A_q$, is nonempty. Observe that this leads to a contradiction since a coloring of $v_q$ with any $c\in A_q$ can be extended to an $A_\psi$-CBC-7-coloring of $(H,P')$ by the definition of extendable reduction. For each $i \in \{1,\ldots,q\}$ we write $R_i$ as $((H_i,P_i),A_i)$. Observe that $P_i = (v_i,\ldots,v_q)$ and that $H_i = H[\{v_i,\ldots,v_q\}]$, and denote by $\ell_i(v)$ the value $\lvert A_i(v)\rvert$, for each $v\in \{v_i,\ldots,v_q\}$. 
In order to obtain the desired sequence of extendable reductions, we want to apply Lemma~\ref{2reduction}. For this, we need to ensure that, at the beginning and after each step $i$ of the procedure, the inequalities below hold. 

\begin{equation}\ell_i(v_j) \ge d_{H_i}(v_j) + 2\mbox{, for every $j$ such that $i<j<q$.} \label{eq1}\end{equation}

\begin{equation}\ell_i(v_i) \ge d_{H_i}(v_i)+1\mbox{, if $i<q$}.\label{eq2}\end{equation} 

\begin{equation}
\ell_i(v_q)\ge
\left\{\begin{array}{ll}
d_{H_i}(v_q)+2 & \mbox{, if $i<q$}\\
                      1 & \mbox{, otherwise}
\end{array}\right.\label{eq3}\end{equation}

\begin{claim}
If Inequalities~(\ref{eq1}),~(\ref{eq2}), and~(\ref{eq3}) hold for $R_i$, with $1\le i<q$, then $R_i$ has an extendable reduction on $v_i$.
\end{claim}
\begin{claimproof}
Because $d(v_j)\le 5$, for every $v_j\in V(P')$, and by Inequality (\ref{eq2}), we get that Conditions (\ref{1}) and (\ref{2}) of Lemma \ref{2reduction} hold. Now, suppose that $d_{H_i}(v_i) = 4$. Recall that $d_G(v_1)\le 4$; hence $1<i<q$, which implies that $d_H(v_i) = 5$. But since $d_{H_i}(v_i)=4$, this means that $N_G(v_i) = N_{H_i}(v_i)\cup\{v_{i-1}\}$, which implies that $A_{i-1}(v_i) = [7]$. Then, Condition (\ref{3}) follows by the definition of reduction. 
\end{claimproof}

\vspace{1mm}
We first argument that these inequalities initially hold. Recall that $H_1=H$, $P_1=P'$, and $A_1=A_\psi$.
First, consider any $j\in \{2,\ldots,q-1\}$. Since $F$ is a linear tree, we have that $N_F(v_j)\subseteq P'$, which means that $\ell_1(v_j)\ge 7 - d_{G-H}(v_j) = 7 - (d_G(v_j) - d_H(v_j))$. By the choice of $v_1$ and $v_q$, we know that $d_G(v_j)= 5$, which in turn implies Inequality~(\ref{eq1}).
Now, by Lemma~\ref{lem:degree}, we know that $d_P(v_1) = 2$; so let $v\in N_P(v_1)\setminus\{v_2\}$. Note that $v$ forbids~3 colors for $v_1$, while each other colored neighbor of $v_1$ forbids just one color. This gives us that $\ell_1(v_1) \ge 7 - (d_{G-H}(v_1) + 2d_{P-P'}(v_1)) = 5 - (d_G(v_1) - d_H(v_1))\ge d_H(v_1)+1$. Analogously, for $v_q$ we get: if $d(v_q) = 3$, then $d_F(v_q) = 2$ and $\ell_1(v_q) \ge d_H(v_q)+2$; and if $v_q$ is a leaf in $P$, then by Lemma~\ref{lem:degree} we get $d_G(v_q)=5$, and as before $\ell_1(v_q) = 7 - d_{G-H}(v_q) \ge d_H(v_q)+2$.

Now, suppose that we are at step $i$ of our construction, $1\le i < q$, and let $R_{i+1}$ be an extendable reduction of $R_i$. We want to prove that Inequalities~(\ref{eq1}), (\ref{eq2}), and (\ref{eq3}) also hold for $R_{i+1}$.
First, note that if $v_j\in N(v_i)\setminus\{v_{i+1}\}$, then both $d_{H_{i+1}}(v_j)$ and $\ell_{i+1}(v_j)$ decrease by exactly 1; hence, Inequality~(\ref{eq1}) holds, as well as Inequality~(\ref{eq3}) in the case where $i<q-1$. Similarly, $d_{H_{i+1}}(v_{i+1})$ decreases by~1, while $\ell_{i+1}(v_{i+1})$ decreases by at most~2; hence, if $i<q-1$, we have that $\ell_i(v_{i+1})\ge d_{H_i}(v_{i+1})+2$, which means that Inequality~(\ref{eq2}) also holds for $R_{i+1}$. Finally, suppose that $i=q-1$. Then $\ell_{q-1}(v_q) \ge d_{H_{q-1}}(v_q)+2 = 3$, and by the definition of reduction we get that $\ell_q(v_q)\ge 1$, i.e., Inequality~(\ref{eq3}) holds also when $i=q-1$, and we are done proving~($\ref{FS1}$) and~($\ref{FS3}$).

Finally, in order to prove ($\ref{FS2}$), suppose that $d(v)\ge 4$, for every $v\in V(P)$, and let $u,v,w\in V(P)$ be the closest three vertices of degree 4 in $P$, where $v$ is between $u$ and $w$.  Write the subpath of $P$ between $u$ and $w$ as $P'=(v_1=u,v_2,\ldots,v_q=w)$ and let $v_p=v$. Denote $G[V(P')]$ by $H$, and let $\psi$ be a CBC-7-coloring of $(G-H,F-H)$. Note that:

\begin{itemize}
\item For each $z\in V(P')\setminus \{u,v,w\}$, we get $a_\psi(z) \ge 7 - d_{G-H}(z) = 7 - (d_G(z)-d_H(z)) = 2+d_H(z)$; 
\item For $z\in\{u,w\}$, we get $a_\psi(z) \ge 4 - (d_{G-H}(z)-1) = 1+d_{H}(z)$; and
\item $a_\psi(v) = 7 - d_{G-H}(v) = 3 + d_{H}(v)$.
 \end{itemize}

By arguments similar to the ones made for the first two cases, one can verify that a series of extendable reductions can be made on $P'$, from $v_1$ up to $v_{p-1}$, and from $v_q$ down to $v_{p+1}$, until we end up with just $v_p$ with non-empty list. 
\end{proof}


\subsection{Discharging Method}

In this section, we finish the proof of Theorem~\ref{main:forest}. For this, we use a definition similar to the one used in the proof of Theorem~\ref{thm:noadjacentC3}. We make an abuse of language and use the same nomenclature. 
%
Consider a plane graph $G$ and its dual $G^*$, and let $F_3$ be the set of faces of degree~3 in $G$ (alternatively, the set of vertices of degree~3 in $G^*$). We denote the graph $G^*-F_3$ by $G^*_5$, and say that a component of $G^*_5$ is an \emph{island of $G$}. Also, if $H$ is an acyclic component of $G^*_5$ such that $d_{G^*}(f) = 5$, for every $f\in V(H)$, then we say that $H$ is a \emph{bad island of $G$}. 
We denote the set of bad islands of $G$ by $\Gamma$ and we let $\gamma$ denote $\lvert \Gamma\rvert$. 
Also, for $v\in V(G)$, we denote by $\Gamma(v)$ the set of bad islands containing $v$, and by $\gamma(v)$ the value $\lvert \Gamma(v)\rvert$. If $X\subseteq V(G)$, then $\Gamma(X) = \bigcup_{x\in X}\Gamma(x)$, and $\gamma(X) = \lvert \Gamma(X)\rvert$. In the remainder of the text, although we refer to $G$ as being planar, we are implicitly considering a planar embedding of $G$ and its islands.

\begin{lemma}\label{mainlemma}
Let $G$ be a planar graph without cycles of lenght 4 as subgraph. Then, \[\lvert E(G)\rvert \le 2\lvert V(G)\rvert-4+\frac{\gamma}{3}.\]
\end{lemma}
\begin{proof}
Let $f_3,f_5$ denote the number of faces of degree 3 and 5, respectively, and let $\lvert E(G)\rvert$ be denoted by $m$. Also, denote by ${\cal F}$ the set of faces of $G$ and by $\lvert f\rvert$ the degree of a face $f\in {\cal F}$. We claim that: \begin{equation}\label{claim} 3f_3+f_5\le m+\gamma\end{equation}
This implies that $t = \sum_{f\in {\cal F}}(\lvert f\rvert - 6) \ge -3f_3-f_5\ge -m-\gamma$. On the other hand $t = \sum_{f\in {\cal F}}(\lvert f\rvert) - 6\lvert {\cal F}\rvert = 2m - 6\lvert {\cal F}\rvert$. Combining these and applying Euler's Formula we get (below, $n$ denotes $\lvert V(G)\rvert$):
$$ 2m - 6(2-n+m)\ge -m-\gamma \Longleftrightarrow m\le 2n-4+\frac{\gamma}{3}$$

It remains to prove Inequality~\ref{claim}. For this, we partition $E(G)$ in $E_3,\overline{E}_3$, where $E_3$ is described below and $\overline{E}_3 = E(G)\setminus E_3$.
$$E_3 = \{e\in E(G)\mid \mbox{ $e$ is contained in some face of degree 3}\}.$$

Because $G$ has no cycle of length 4, we trivially get that $\lvert E_3\rvert = 3f_3$. We prove that $\lvert \overline{E}_3\rvert \ge f_5 - \gamma$, thus finishing the proof. For this, note that if $e\in \overline{E}_3$, then there is an edge $e^*$ in $G^*_5$ related to $e$. On the other hand, if $e^*\in E(G^*_5)$, then $e^*$ is related to an edge $e\in E(G)$ that separates faces of degree at least 5; hence, $e\in \overline{E}_3$. Therefore, $\lvert \overline{E}_3\rvert = \lvert E(G^*_5)\rvert$. Finally, because the number of edges in any graph is at least the number of vertices minus the number of acyclic components of the graph, we get:
$$\lvert \overline{E}_3\rvert \ge \lvert V(G^*_5)\rvert - \gamma \ge f_5-\gamma.$$
\end{proof}

Supposing that $(G,F)$ is a minimal counterexample to Theorem~\ref{main:forest}, we apply the discharging method to prove that $\sum_{v\in V(G)}d(v) \ge 4\lvert V(G)\rvert + \frac{2\gamma}{3}$, contradicting Lemma~\ref{mainlemma}.
For this, we start by giving charge $d(v)-4$ to every $v\in V(G)$, and charge $-2/3$ to every bad island. The discharging rules ensure that every vertex and every bad island end up with a non-negative charge (i.e., Property \ref{property} below holds), which clearly contradicts Lemma~\ref{mainlemma}. The rules are applied in the order they are presented. Also, given $x\in V(G)\cup \Gamma$, the initial charge of $x$ is denoted by $\mu_0(x)$, and the charge of $x$ after Rule $i$ is applied is denoted by $\mu_i(x)$, for each $i\in \{1,\ldots, 5\}$. 

\begin{property}\label{property}
 After Rule $i$ is applied, we have that $\mu_i(v)\ge 0$ and $\mu_i(b)\ge 0$, for every vertex $v$ iterated in Rule~$i$ and every bad island $b$ containing $v$.
 \end{property}

The proof following each rule is a proof that Property~\ref{property} holds after the corresponding rule has been applied.

\setcounter{Drule}{0}
\begin{Drule}\label{rule1}
For every $v\in V(G)$ with $d(v)\ge 6$, send $2/3$ from $v$ to each $b\in\Gamma(v)$.
\end{Drule}
\begin{proof}
Consider $v\in V(G)$ with $d(v)\ge 6$. Because every island containing $v$ receives $2/3$, we just need to prove that $\mu_1(v)\ge 0$. Because $G$ has no cycles of length~4, observe that $\gamma(v)\le \frac{d(v)}{2}$. This gives us that: 
\begin{equation}\label{sobra}
\mu_1(v)\ge d(v) - 4 - \frac{2}{3}\gamma(v)\ge d(v) - 4 - \frac{2}{3}\cdot\frac{d(v)}{2}\ge \frac{2}{3}d(v) - 4 \ge 0.
\end{equation}
\end{proof}

The following proposition will be useful in the remainder of the text. Observe that it holds because at least one face containing $uv$ cannot be a a face of degree~3, as otherwise we get a cycle of length~4.

\begin{proposition}\label{prop:gamma}
If $G$ is a graph without cycles of length~4, and $uv\in E(G)$, then there exists a face of degree greater than~3 containing $uv$.
\end{proposition}

\begin{Drule}\label{rule2}
Let $P = (v_1,\ldots,v_q)$ be a heavy subpath containing no vertex with degree smaller than 5. We have the following cases:
\begin{itemize}
\item[R2.1]\label{R21} If $P$ is a component of $F$, send charge $2/3$ from $\mu_1(v_1)+\mu_1(v_2)$ to every $b\in \Gamma(\{v_1,v_2\})$. After this, if $q\ge 3$, then for each $i\in\{3,\ldots,q\}$, send charge $2/3$ from $v_i$ to  $b\in \Gamma(v_i)\setminus\Gamma(v_{i-1})$, . 

\item[R2.2]\label{R22} Otherwise, let $v_0\in N_F(v_1)\setminus \{v_2\}$. For every $i\in\{1,\ldots,q\}$, send charge $2/3$ from $v_i$ to $b\in \Gamma(v_i)\setminus\Gamma(v_{i-1})$. 
\end{itemize}
\end{Drule}
\begin{proof}
First, note that $\mu_1(v_i)=1$, for every $i\in\{1,\ldots,q\}$. Suppose that $P$ is a component of $F$. Note that Lemma~\ref{lem:totaldegree} implies that $q\ge 2$. By Proposition~\ref{prop:gamma}, we get that $\gamma(\{v_1,v_2\})\le 3$, and that, when $q\ge 3$, then for every $i\in \{3,\ldots, q\}$ we get $\lvert \Gamma(v_i)\setminus\Gamma(v_{i-1})\rvert\le 1$. Property \ref{property} follows.

Now, suppose that $P$ is not a component of $G$, in which case we can suppose, without loss of generality, that $v_0$ exists. By the definition of heavy path, we know that $d(v_0)\ge 6$, which, by Rule \ref{rule1}, implies that the island in $\Gamma(v_0)\cap \Gamma(v_1)$ has non-negative charge. Now, applying Proposition~\ref{prop:gamma}, for each $v_i\in V(P)$ we get that $\lvert \Gamma(v_i)\setminus\Gamma(v_{i-1})\rvert \le 1$. Hence, Property \ref{property} follows.
\end{proof}

\begin{Drule}\label{rule3}
Let $P = (v_1,\ldots,v_q)$ be a heavy subpath containing exactly one vertex with degree smaller than 5, namely $v_p$, and let $v_0\in N_F(v_1)\setminus P$ and $v_{q+1}\in N_F(v_q)\setminus P$. We have the following cases.
 \begin{itemize}
\item[R3.1]\label{R31} If $q\ge 2$, we can suppose that $p<q$. 
\begin{enumerate}[(i)]
 \item Send charge $2/3$ from $v_i$ to $b\in\Gamma(v_i)\setminus\Gamma(v_{i-1})$, for each $i\in\{1,\ldots,p-1\}$;
 \item Send charge $2/3$ from $v_i$ to $b\in\Gamma(v_i)\setminus\Gamma(v_{i+1})$, for each $i\in\{p+2,\ldots,q\}$;
 \item If $d(v_p)=3$, then $v_{p+1}$ sends charge $1$ to $v_p$. Otherwise, $v_{p+1}$ sends charge $2/3$ to $b\in \Gamma(v_p)\cap \Gamma(v_{p+1})$.
\end{enumerate}

\item[R3.2] \label{R32}If $q=1$ and $d(v_1)=3$, let $b\in \Gamma(v_1)$. Send charge 1 from $\mu_2(v_0) +\mu_2(v_2)+\mu_2(b)$ to $v_1$. 
\end{itemize}
\end{Drule}
\begin{proof}
By Lemma \ref{lem:forbiddenStructures}, we know that $v_0$ and $v_{q+1}$ exist, and, by Rule \ref{rule1}, we know that the islands in $\Gamma(v_0)\cap \Gamma(v_1)$ and $\Gamma(v_q)\cap \Gamma(v_{q+1})$ have non-negative charge. First, suppose that $q\ge 2$. By arguments similar to the ones in the previous demonstrations, one can see that the vertices in $\{v_1,\ldots,v_{p-1},v_{p+2},\ldots,v_q\}$, as well as the islands containing them, have non-negative charge. Also, note that, by Proposition \ref{prop:gamma}, either $d(v_p)=3$ and the only island containing $v_p$ also contains $v_{p-1}$ and $v_{p+1}$, or $d(v_p) = 4$ and the island in $\Gamma(v_p)\cap \Gamma(v_{p+1})$ is the only one that 	might not be satisfied yet. In either case, one can verify that the rule satisfies $v_p$ or the refered island, depending on the case. 

Now, suppose that $q=p=1$. If $d(v_1)=4$, then $\Gamma(v_1)\subseteq \Gamma(v_0)\cup\Gamma(v_2)$ and nothing needs to be done; so suppose otherwise. 
First note that, because $d(v_1)=3$, the island $b\in \Gamma(v_1)$ also contains $v_0$ and $v_2$. This means that $b$ has received charge from both $v_0$ and $v_2$ when Rule~\ref{rule1} is applied; hence $\mu_2(b)=2/3$. We end the proof by showing that $\mu_2(v_2)=\mu_1(v_2)\ge 2/3$. 
Note that, since $d(v_1)=3$ and because $G$ has no cycle of length 4, we can suppose that $v_1$ has no common neighbor with $v_2$. 
Therefore, if $d(v_2) = 6$, then $\gamma(v_2)=2$, and applying the first part of Inequality~\ref{sobra}, we get that $\mu_2(v_2) = 6-4-4/3=2/3$. On the other hand, if $d(v_2)\ge 7$, we get $\mu_2(v)\ge 2/3$ by Inequality~\ref{sobra}.
 \end{proof}

In the next discharging rule, given $X\subseteq V(G)$, we denote $\sum_{v\in X}\mu_3(x)$ by $\mu_3(X)$.

\begin{Drule}\label{rule4}
Let $P = (v_1,\ldots,v_\ell)$ be a heavy subpath containing exactly two vertices with degree smaller than 5, namely $v_p$ and $v_q$, $p<q$. Let $v_0\in N_F(v_1)\setminus P$ and $v_{\ell+1}\in N_F(v_\ell)\setminus P$. Define
$$\beta = \Gamma(V(P))\setminus \Gamma(\{v_0,v_{\ell+1}\})\mbox{, and}$$
$$\mu = \mu_3(V(P)) + \frac{2}{3}\lvert \Gamma(v_0)\cap \Gamma(v_{\ell+1})\rvert.$$

If $\mu \ge \frac{2}{3}\cbeta$, then send $2/3$ from $V(P)$ and $\Gamma(v_0)\cap \Gamma(v_{\ell+1})$ to each $b\in \beta$.
\end{Drule}

By the condition under which it is applied, Rule \ref{rule4} clearly satisfies Property~\ref{property}. However, we still need a final rule for the paths on which the condition $\mu\ge \frac{2}{3}\cbeta$ does not hold. Before we present the rule, we give sufficient conditions for Rule \ref{rule4} to be applied.

\begin{lemma}\label{lem:rule4}
If $P$ is a heavy subpath containing exactly two vertices with degree smaller than~5, and either $\lvert V(P)\rvert \ge 4$, or $\gamma(V(P))\le \lvert V(P)\rvert$, then $\mu\ge\frac{2}{3}\cbeta$.
\end{lemma}
\begin{proof}
Consider $P,v_p,v_q,v_0,v_{\ell+1},\beta,\mu$ be all defined as in Rule \ref{rule4} (recall that $v_0,v_{\ell+1}$ exist by Lemma~\ref{lem:forbiddenStructures}). First note that 
$$\lvert \beta \rvert = \gamma(V(P)) - \lvert \Gamma(V(P))\cap \Gamma(\{v_0,v_{\ell+1}\}\rvert.$$ 
Also, by Proposition~\ref{prop:gamma}, we have 
$$\gamma(V(P))\le 2\ell - (\ell-1) = \ell+1.$$ 
Finally, by Lemma~\ref{lem:forbiddenStructures}, we get that $d(v_p)=d(v_q)=4$, and $d(v_i)=5$, for every $v_i\in V(P)\setminus\{v_p,v_q\}$. Hence $$\mu_3(V(P))=\ell-2.$$

Now, denote by $t$ the value $\lvert \Gamma(V(P))\cap \Gamma(\{v_0,v_{\ell+1}\})\rvert$. By Proposition \ref{prop:gamma}, we know that $t\ge 1$. We analyse the following cases:
\begin{itemize}
 \item If $t=1$, then the islands in $\Gamma(v_0)\cap \Gamma(v_1)$ and $\Gamma(v_\ell)\cap \Gamma(v_{\ell+1})$ must be the same, i.e., $\Gamma(v_0)\cap \Gamma(v_{\ell+1})\neq \emptyset$, and $\cbeta = \gamma(V(P))-1$. Therefore, $$\mu \ge \mu_3(V(P))+\frac{2}{3} = \ell-2+\frac{2}{3} = \ell - \frac{4}{3}.$$
  If $\ell\ge 4$, then $\cbeta\le \ell$ and $\mu\ge \ell-\frac{4}{3}\ge \frac{2}{3}\ell\ge \frac{2}{3}\cbeta$. And if $\gamma(V(P))\le \ell$, then $\cbeta\le \ell-1$, and, since $\ell\ge 2$, we get $\mu = \ell - \frac{4}{3} \ge \frac{2}{3}(\ell - 1)\ge \frac{2}{3}\cbeta$.
 
 \item Now, if $t\ge 2$ and $\ell\ge 4$, then $\cbeta \le \ell-1$, and $\mu\ge \ell-2\ge \frac{2}{3}(\ell-1)\ge \frac{2}{3}\cbeta$. Finally, if $t\ge 2$ and $\gamma(V(P))\le \ell$, then $\cbeta\le \ell-2$ and clearly $\mu\ge \ell-2\ge \cbeta\ge \frac{2}{3}\cbeta$.
\end{itemize}
\end{proof}

Now, consider $P$ as in Rule \ref{rule4} and suppose that the rule is not applied, which means that there might still exist some bad island intersecting $V(P)$ with negative charge. If such an island exists, we call such a path \emph{defective}. Before we present the last discharging rule, we need the lemmas below. We mention that by Lemma~\ref{lem:rule4}, if $P$ is defective then $\ell \le 3$ and $\gamma(V(P))\ge \ell+1$, where $\ell=\lvert V(P)\rvert$.

\begin{lemma}\label{lem:defectivepath}
Let $P$ be a defective path of size $\ell$ with extremities $v_1$ and $v_\ell$, and denote by $v_2$ the neighbor of $v_1$ in $P$ (hence, it might happen that $\ell=2$). Also, let $v_0\in N_F(v_1)\setminus\{v_2\}$, and $v_{\ell+1}\in N_F(v_\ell)\setminus\{v_{\ell-1}\}$.
Then, for each $i\in \{1,2,\ell\}$, we have that $v_i$ is contained in exactly two bad islands (which means that $v_i$ is contained in two 3-faces that separate these bad islands), and $v_{i-1}v_{i+1}\notin E(G)$. 
\end{lemma}
\begin{proof}
First, suppose that $i\in\{1,2,\ell\}$ is such that $v_i$ is contained in at most one triangle, which means that $\gamma(v_i)\le 1$. Note that if $\ell = 3$, then $\lvert \Gamma(v_1)\cap\Gamma(v_2)\cap\Gamma(v_3)\rvert - \lvert \Gamma(v_1)\cap\Gamma(v_3)\rvert\le 0$. This justifies the second line in the equation below.
$$\begin{array}{rl}
\gamma(V(P)) = &   \lvert \bigcup_{v_j\in V(P)}\Gamma(v_j) \rvert \\
                        \le & \sum_{j\in\{1,2,\ell\}}\gamma(v_j) - \sum_{j\in\{1,\ell-1\}} \lvert \Gamma(v_j)\cap\Gamma(v_{j+1})\rvert\\
                        \le & \sum_{j\in\{1,2,\ell\}\setminus\{i\}} \gamma(v_j) + \gamma(v_i) - (\ell-1) \\
                        \le & 2(\ell - 1) + 1 - \ell + 1 = \ell
\end{array}$$

This means that $P$ satisfies Lemma \ref{lem:rule4}, a contradiction. Note also that this actually implies that each $v_i$ is contained in exactly two bad islands.

Now, suppose that $i\in\{1,2,\ell\}$ is such that $v_{i-1}v_{i+1}\in E(G)$. Note that if $\ell=3$ and $i=2$, then $\gamma(V(P)) = \gamma(\{v_1,v_3\})$, and the island in $\Gamma(v_0)\cap \Gamma(v_1)$ also contains $v_3$. This implies that $\gamma(V(P)) = 3$, contradicting Lemma~\ref{lem:rule4}. 
So suppose, without loss of generality, that $i=1$ and let $b$ be the island containing $v_0v_2$. Note that $\Gamma(v_1)\subseteq \Gamma(\{v_0,v_2\})$; therefore, $\beta = \Gamma(\{v_2,v_\ell\})\setminus \Gamma(\{v_0,v_{\ell+1}\})$. First consider $\ell=2$. If $b$ also contains $v_3$, then $\cbeta \le \lvert\Gamma(v_2)\setminus \{b\}\rvert = 1$, and $\Gamma(v_0)\cap \Gamma(v_3)\neq \emptyset$, which implies $\mu \ge \frac{2}{3}\cbeta$. And if $b$ does not contain $v_3$, then $\Gamma(v_2)\subseteq \Gamma(\{v_0,v_3\})$, in which case $\beta=\emptyset$. Both cases are contradictions. Therefore, suppose that $\ell=3$, and let $B$ denote $\Gamma(\{v_0,v_4\})$. Note that:
$$\begin{array}{rl}
\cbeta = & \lvert \Gamma(\{v_2,v_3\})\setminus B\rvert\\
          = & \lvert (\Gamma(v_2)\setminus B)\cup (\Gamma(v_3)\setminus B)\rvert\\
          \le & \lvert (\Gamma(v_2)\setminus B)\rvert + \lvert (\Gamma(v_3)\setminus B)\rvert \le 2.
\end{array}$$

The last part holds since $b\in \Gamma(v_2)\cap B$, and $\Gamma(v_3)\cap \Gamma(v_4)\neq \emptyset$ (Proposition \ref{prop:gamma}). If $\cbeta\le 1$ we are done since $\mu\ge 1$. Therefore, suppose $\cbeta = 2$, in which case we must have $(\Gamma(v_2)\setminus B)\cap (\Gamma(v_3)\setminus B)=\emptyset$. So, let $b_i\in \Gamma(v_i)\setminus B$, for $i=2$ and $i=3$, and let $b^*\in \Gamma(v_3)\cap\Gamma(v_4)$. Because $\Gamma(v_2)\cap \Gamma(v_3)\neq \emptyset$ and $b_2\neq b_3$, we get $b=b^*$, i.e., $b\in \Gamma(v_0)\cap \Gamma(v_4)$. Therefore, we get $\mu\ge 1+ \frac{2}{3} > \frac{4}{3} = \frac{2}{3}\cbeta$, a contradiction.
\end{proof}

The next lemma is the final step before we can present the last discharging rule. We denote by $\Theta$ the set of bad islands with negative charge, and by $D$ the set of vertices of degree 5 which are contained in some island in $\Theta$.

\begin{lemma}
Let $b\in \Theta$, and $f$ be a face of degree~5 in $b$. Then $f$ contains at least one vertex of $D$ and, if it contains exactly one such vertex, namely $u$, then  $b$ is the only island in $\Theta$ that contains $u$.
\end{lemma}
\begin{proof}
Let $f=(v_1,\ldots,v_5)$ be such that $v_i$ is contained in some defective path, for each $i\in\{1,\ldots,5\}$. Without loss of generality, suppose that $d(v_i)=4$, for every $i\in\{1,\ldots,4\}$. First, we want to prove that $(v_1,\ldots,v_5)$ is an induced cycle in $G$. So suppose that $v_1v_3\in E(G)$. Since $f$ is a 5-face in $G$, we must have that the edge $v_1v_3$ is traced in the outer side of $f$. Because $\delta(G)\ge 3$, one can verify that this implies that $(v_1,v_2,v_3)$ is not a 3-face in $G$, which in turn implies that $v_1$ is contained in at most one bad island,  contradicting Lemma \ref{lem:defectivepath}. Observe that the same argument can be applied to conclude that $v_iv_j\notin E(G)$, for every $i\in\{1,\ldots,4\}$ and every $j\in \{1,\ldots,5\}\setminus\{i\}$. Now observe that, by Lemma~\ref{lem:defectivepath}, there must exist $u_1,\ldots,u_5$, where $u_5\in N(v_1)\cap N(v_5)$, and $u_i\in N(v_i)\cap N(v_{i+1})$, for each $i\in \{1,\ldots,4\}$. This means that every island in $\Theta$ is a face of degree 5. We claim that $d(v_5) = 5$. Supposing it holds, let $w\in N(v_5)\setminus\{v_1,v_4,u_4,u_5\}$; also let $f_1$ be the face containing $u_4v_5$ different from $(v_4,v_5,u_5)$, and $f_2$ be the face containing $u_5v_5$ different from $(u_5,v_5,v_1)$. Because $G$ has no cycles of length~4, we know that $f_1$ and $f_2$ have degree bigger than~3, and that they share the edge $v_5w$. This means that $f_1$ and $f_2$ are within the same island $t$, which implies that $t\notin\Theta$, and the lemma follows, i.e., $b$ is the only island in $\Theta$ containing $u$. It remains to prove our claim.

Suppose by contradiction that $d(v_5)=4$, and let $H$ denote the induced subgraph $G[\{v_1,\ldots,v_5,u_1,\ldots,u_5\}]$. Because $d_F(v_i)=2$ and $N(v_i)\subseteq V(H)$, for every $i\in \{1,\ldots,5\}$, we know that $H$ must contain every edge in $F$ incident to $\{v_1,\ldots,v_5\}$. For each $v_i$, let $E_i$ denote the set $\{uv_i\in E(F)\}$; we know that $\lvert E_i\rvert =2$. Therefore, if $E_i\cap E_j=\emptyset$, for every $i,j\in \{1,\ldots,5\}$, $i\neq j$, then $\lvert E(H)\cap E(F)\rvert = \lvert\bigcup_{i=1}^5E_i \rvert=\sum_{i=1}^5\lvert E_i\rvert = 10 = \lvert V(H)\rvert$, contradicting the fact that $F$ is acyclic. We can then suppose, without loss of generality, that $v_1v_2\in E(F)$. By Lemmas~\ref{lem:forbiddenStructures} and~\ref{lem:defectivepath}, we get that $\{u_5v_1,u_2v_2\}\subseteq E(F)$. Also, by Lemma \ref{lem:defectivepath}, we get $\lvert \{v_3v_4,v_3u_3\}\cap E(F)\rvert\le 1$ and $\lvert \{v_4v_5,u_4v_5\}\cap E(F)\rvert\le 1$. This implies that $\{u_5v_5,u_2v_3\}\subseteq E(F)$. It is easy to verify that no matter the choice of edges in $E_4$, we get a cycle in $F$, a contradiction.
\end{proof}

The lemma above implies the correctness of our final discharging rule.

\begin{Drule}
Let $K = (D,E)$ be such that $uv\in E$ if and only if $u$ and $v$ are within the same bad island $b\in \Theta$. For each component $K'$ of $K$, apply one of the following:
\begin{enumerate}
\item[R5.1] If $\lvert V(K')\rvert \ge 2$, let $T$ be a spanning tree of $K'$ and let $uv\in E(T)$. Send charge $2/3$ from $\{u,v\}$ to each island in $\Gamma(\{u,v\})$, and for every $w\in V(T)\setminus\{u,v\}$, send charge $2/3$ from $w$ to the island in $\Gamma(w)\setminus\Gamma(w')$, where $w'\in N_T(w)$ separates $w$ from $uv$.

\item[R5.2] If $V(K') = \{u\}$, send $2/3$ from $u$ to the bad island in $\Theta$ containing $u$.
\end{enumerate}
\end{Drule}

%
%

\bibliographystyle{plain}
\bibliography{linearForest-v2}

\end{document}